\documentclass[11pt,a4paper,reqno]{amsart}
\usepackage{graphicx}%
\usepackage{multirow}%
\usepackage{amsmath,amssymb,amsfonts,mathtools,bm}%
\usepackage{amsthm}%
\usepackage{mathrsfs}%
\usepackage[title]{appendix}%
\usepackage{xcolor}%
\usepackage{textcomp}%
\usepackage{setspace}
\usepackage{manyfoot}%
\usepackage{booktabs}%
\usepackage[ruled]{algorithm2e}%
\usepackage{algorithmicx}%
\usepackage{algpseudocode}%
\usepackage{listings}%
\usepackage{booktabs}
\usepackage{pifont}
\usepackage{subcaption}%
\usepackage{makecell}
  \usepackage[top=1.5in, bottom=1.2in,
  left=1in, right=1in]{geometry}
\usepackage{hyperref}
\hypersetup{colorlinks = true, urlcolor = green, linkcolor = red,
  citecolor = blue }
\usepackage[round,authoryear]{natbib}
\usepackage{hyperref}
\hypersetup{
  colorlinks = true,
  urlcolor = green,
  linkcolor = red,
  citecolor = blue
}
\newcommand{\RV}{{\tt RV}}
\newcommand{\LDP}{\mathsf{LDP}}

\newcommand{\mv}[1]{\boldsymbol{#1}}

\newcommand{\qq}{\boldsymbol{q}}

\newcommand{\xx}{\boldsymbol{x}}
\newcommand{\XX}{\boldsymbol{\xi}}
\newcommand{\YY}{\boldsymbol{Y}}
\newcommand{\ZZ}{\boldsymbol{Z}}
\newcommand{\yy}{\boldsymbol{y}}
\newcommand{\zz}{\boldsymbol{z}}

\newtheorem{assumption}{Assumption}%[section]
\newtheorem{illustration}{Numerical Illustration}%[section]
\newtheorem{theorem}{Theorem}[section]
\newtheorem{definition}{Definition}%[section]
\newtheorem{example}{Example}%[section]
\newtheorem{lemma}{Lemma}[section]
\newtheorem{proposition}{Proposition}[section]

\theoremstyle{thmstyleone}%
\theoremstyle{thmstyletwo}%

\theoremstyle{thmstylethree}%

\title[Generating Plausible Stress Scenarios]{Generating Plausible Stress Scenarios via Large Deviations}
    \author{{\large A\MakeLowercase{nand}
    D\MakeLowercase{eo}}}
  \address{Indian Institute of Management, Bannerghatta Road, Bangalore 560076}
  \email{anand.deo@iimb.ac.in}
\begin{document}
\maketitle
\begin{abstract}
Financial stress tests based on handpicked scenarios can mislead risk management by overlooking genuinely dangerous configurations or overemphasizing shocks that are too implausible to be decision-relevant. We develop a systematic method for generating plausible stress scenarios for financial losses driven by exogenous risk factors. The method exploits a large-deviations principle: conditional on a large loss, the risk factors concentrate near the most likely stress configurations. We use this structure to define representative stress distributions and to extrapolate observed samples into more extreme scenarios while preserving the relative plausibility of stress mechanisms. As a result, the procedure can generate informative stress scenarios even when historical data contain few or no observations in the stressed regime. Numerical experiments on two financial network models show that the method recovers the stressed loss law and key stress diagnostics, including in settings where benchmark generators fail to generate any stressed samples.
\end{abstract}

\section{INTRODUCTION}\label{sec:intro}
Crises such as the 2008 global financial crisis and the COVID-19 pandemic have made stress testing a central tool for assessing the resilience of financial systems. In practice, supervisory stress tests specify shocks to underlying risk factors and propagate them through a portfolio or network model. The central difficulty is that useful stress scenarios must be both \emph{severe} and \emph{plausible}. A shock that is mild across all risk factors may leave the system within its normal 
operating range, whereas a handpicked combination of extreme market conditions may be so implausible that it obscures 
rather than clarifies systemic risk. This difficulty is especially acute in multivariate settings, where the joint configurations driving large losses are rarely, if ever, observed in historical data.

A useful theory of stress testing must therefore identify the rare joint configurations that drive systemic losses. This naturally leads to a large-deviations viewpoint: in stressed regimes, plausibility should be tied to the configurations under which the rare event is most likely to occur. What matters is not only the size of individual shocks, but how they combine to produce large losses. Ad hoc scenario design offers little guidance here, and the problem is not resolved simply by replacing such scenarios with a fitted multivariate model. Even models that match observed marginals and correlations may still route exceedances through structurally incorrect parts of the tail, producing severe scenarios that misrepresent the mechanisms of systemic stress. These observations lead to two central questions:
\begin{enumerate}
    \item[\textbf{(Q1)}] Can the probabilistic structure of the risk factors be used to characterize precisely what it means for a stress scenario to be plausible?
    \item[\textbf{(Q2)}] Can this characterization be turned into a data-driven extrapolation procedure that generates plausible stress scenarios even beyond the observed severity range?
\end{enumerate}

In this paper, we provide an answer to the above questions in two parts. { First, we define plausibility endogenously, meaning that the stress configurations are determined by the baseline tail model itself rather than imposed externally through a handpicked shock set.} Second, we use this characterization not merely to rank or locate adverse scenarios, but to construct a benchmark for generating stressed realizations whose exceedances concentrate on those same configurations. Our contributions are below.

\noindent\textbf{(i) A large-deviations characterization of stress scenarios:}
We use large-deviations theory to formalize plausibility in stressed regimes. Conditional on a high-loss event $\{L(\XX)>u\}$, the risk factors concentrate near a most-likely set of configurations that captures the dominant mechanisms by which systemic stress occurs. This yields an endogenous benchmark: a stress distribution is representative if, conditional on exceedance, it concentrates on the same most-likely configurations as the baseline model.

\smallskip
\noindent\textbf{(ii) Data-driven generation of plausible stress scenarios:}
This characterization leads to a natural generator: a sample-based extrapolation procedure that pushes observed samples to more extreme levels while preserving the configurations that drive large losses. Given $n$ baseline samples and an extrapolation factor $s>1$, the procedure generates realizations whose exceedances approximate the target conditional distribution. We show that representative scenarios can be generated for stress events with probabilities of order $n^{-\gamma(s)}$ for some $\gamma(s)>1$, including regimes in which the event is not directly observed in the available sample. Numerical experiments on two financial network models show that the method recovers the stressed loss law and key stress diagnostics in practice.

Our objective is not unbiased rare-event estimation, but the generation of scenarios that recover the stressed regime. At the same time, because the generated scenarios preserve the relevant tail structure, they can also be used to approximate tail risk quantities at the large-deviations scale.

\noindent\textbf{Comparison with Existing Literature:}
There is a large literature on stress testing and scenario-based risk assessment, including foundational frameworks for supervisory and portfolio stress testing (for instance, \cite{morgan1997creditmetrics,berkowitz1999coherent,kupiec2002stress}) and later work on banking-system and systemic-risk stress tests (see \cite{amini2012stress,schuermann2014stress,acharya2014testing}). Closest to our setting are approaches that formalize plausibility or identify likely drivers of stress. One line of work defines plausibility through neighborhoods around a baseline model: \cite{breuer2009find} identify severe scenarios within a local neighborhood, while \cite{breuer2012systematic} implement this idea in a credit-risk setting. A related literature studies reverse stress testing, which seeks the most likely scenarios leading to a prescribed adverse outcome (see \cite{mcneil2012multivariate,kopeliovich2015robust,glasserman2015stress}). More recently \cite{pesenti2019reverse} and \cite{millossovich2024theory} develop theoretical frameworks based on reweighting stress-scenario probabilities.

Much of this literature treats stress testing primarily as an \emph{identification} problem: the goal is to locate adverse configurations, either within an exogenously specified neighborhood of a baseline model or through reverse stress testing tied to a prescribed outcome. Our perspective differs in two respects. First, we define plausibility endogenously through a large-deviations characterization of the baseline tail, thereby identifying the stress mechanisms selected by the model itself. Second, we use this characterization not only to rank or locate adverse scenarios, but to construct a generative benchmark: stressed realizations are representative when their exceedances preserve the same large-deviations structure as the baseline model. Unlike reweighting approaches, which tilt the baseline measure toward an exogenously specified stress event, our method generates realizations directly from the model’s most-likely stress configurations.

{The present paper shares the tail-asymptotic setting of \cite{deo2025achieving}, but repurposes it for data-driven stress-scenario generation. \cite{deo2025achieving} work with a specified input distribution and use self-structuring transformations for importance sampling estimation of the scalar probability $P(L(\XX)>u)$. Here, we start from samples and use the same tail extrapolation idea to construct an empirical distribution of transformed exceedances for the conditional stress regime given $\{L(\XX)>u\}$. Thus, the contribution is not a tail-probability estimator, but a data-driven stress-generation framework whose outputs are assessed by recovery of the baseline large-deviations geometry and stressed-system diagnostics such as severity, breadth, and concentration.
}

\noindent\textbf{Paper Outline:}
Section~\ref{sec:setup} formulates the stress-scenario generation problem precisely. Section~\ref{sec:LDP_Stress} develops the large-deviations framework and the benchmark notion of representative stress generation. Section~\ref{sec:numericals} validates the theory numerically on two financial network examples. Proofs are collected in Section~\ref{sec:proofs}.

\noindent\textbf{Notation:}
Vectors are denoted in boldface; for example, $\zz=(z_1,\ldots,z_d)$. For a random vector $\XX=(\xi_1,\ldots,\xi_d)$, let $\mv\Lambda$ denote the component-wise marginal cumulative hazard transform, defined by $\Lambda_i(t):=-\log \mathbb P(\xi_i>t)$ for $i=1,\ldots,d$. We let $\qq$ denote the component-wise left inverse of $\mv\Lambda$, and define $\qq^\star:=\lim_{t\to\infty}{\qq(t)}/{\|\qq(t)\|_\infty}$,
whenever the limit exists. For $\mv\alpha=(\alpha_1,\ldots,\alpha_d)$, we write $\alpha_\star:=\min_{1\le i\le d}\alpha_i$. For $A\subset\mathbb R^d$, we let $A_\delta:=\{\yy:d(\yy,A)\le \delta\},$
where
$d(\yy,A):=\inf\{\|\yy-\zz\|:\zz\in A\}$
is the distance from $\yy$ to $A$.
\section{TECHNICAL PRELIMINARIES AND PROBLEM SETUP}\label{sec:setup}
Our analysis rests on two ingredients: a marginal standardization that puts the risk factors on comparable tail scales, and a large-scale approximation of the loss function. The first step is to standardize the coordinates so that tail growth is measured on a common scale across risk factors.
\begin{definition}\label{def:RV}
A function $f:\mathbb R_+\to\mathbb R_+$ is said to be \textit{regularly varying} with index $\gamma\in\mathbb R$ if
\[
\frac{f(tx)}{f(t)}\to x^\gamma
\qquad \text{for every } x>0
\]
as $t\to\infty$. We denote this by $f\in\RV(\gamma)$, or simply $f\in\RV$ when the index is not important. 
\end{definition}

{ Regular variation is used to construct the marginal transforms $\Lambda_i$, which put all coordinates on a common tail scale. This standardization isolates joint tail dependence from differences in marginal tail decay, allowing the geometry of rare losses to be studied on the transformed scale.}

\begin{definition}\label{def:LDP}
A family of random vectors $\{\ZZ_u\}_{u>0}$ is said to satisfy a \textit{large deviations principle} with speed $g:\mathbb R_+\to\mathbb R_+$ satisfying $g(u)\to\infty$ and rate function $I^\star$ if
\begin{equation}\label{eqn:LDP_def}
\limsup_{u\to\infty}[g(u)]^{-1}\log \mathbb P(\ZZ_u\in F)
\le -\inf_{\zz\in F} I^\star(\zz)
\quad \text{and} \quad
\liminf_{u\to\infty}[g(u)]^{-1}\log \mathbb P(\ZZ_u\in G)
\ge -\inf_{\zz\in G} I^\star(\zz)
\end{equation}
for every closed set $F$ and open set $G$. In this case, we write $\ZZ_u\in\LDP(g,I^\star)$.
\end{definition}

Following \cite{deValk}, we model tail dependence through an LDP for the marginally standardized vector $\YY=\mv\Lambda(\XX)$. This is the key structural input in our analysis, and the extrapolation procedure is designed precisely to preserve this large-deviations geometry.

\subsection{Problem Setup}
{

Consider a financial system, such as a clearing network or a reinsurance system, whose loss depends on a vector of nonnegative risk factors $\XX\in\mathbb R_+^d$. Let $L(\zz)$ denote the loss incurred in state $\zz$, and assume that $\XX$ admits a density $f_{\XX}$. Fix a high threshold $u$. The benchmark for stress-scenario generation is the conditional distribution $\textrm{Law}(\XX\mid L(\XX)>u)$, supported on the exceedance set $\mathcal S_u=\{\zz\in\mathbb R_+^d:L(\zz)>u\}$. Since direct resampling yields few or no exceedances at such thresholds, our objective is to instead construct a data-driven extrapolation procedure 
to approximate this conditional distribution.
}

% Consider a financial system, such as a clearing network or a reinsurance system, whose loss depends on a vector of nonnegative risk factors $\XX\in\mathbb R_+^d$. Let $L(\zz)$ denote the loss incurred in state $\zz$, and assume that $\XX$ admits a density $f_{\XX}$. Fix a high threshold $u$. The stress-scenario generation problem is to sample from the conditional law $\textrm{Law}(\XX\mid L(\XX)>u)$,
% whose support is the exceedance set $\mathcal S_u=\{\zz\in\mathbb R_+^d:L(\zz)>u\}$. 

Stress-scenario generation is inherently a tail problem: when $u$ is large, the scenarios that dominate $\{L(\XX)>u\}$ are determined by the extreme behavior of $\XX$ and the large-loss geometry of $L(\cdot)$ rather than by behavior in the body of the distribution. To capture this, we impose the following assumptions on the tail behavior of the risk factors and the loss function.

\begin{assumption}[\textbf{Tail behavior of risk factors}]\label{assume:ldp}
Let $\YY=\mv\Lambda(\XX)$ denote the marginally standardized risk vector. Assume that $\YY$ admits a density of the form
\[
f_{\YY}(\yy)=p(\yy)\exp(-\varphi(\yy)),
\]
where
\begin{equation}\label{eqn:ldp_assume}
\lim_{n\to\infty}n^{-1}\varphi(n\yy_n)=\varphi^\star(\yy)
\qquad\text{and}\qquad
\lim_{n\to\infty}n^{-\varsigma}\log p(n\yy_n)=0
\end{equation}
for every sequence $\yy_n\to\yy\neq \mathbf 0$ and every $\varsigma>0$. In addition, there exist positive numbers $\alpha_1,\ldots,\alpha_d$ such that $\Lambda_i\in\RV(\alpha_i)$ for each $i$.
\end{assumption}
Assumption~\ref{assume:ldp} says that after marginal standardization, the transformed risk vector admits a large-deviations description with rate function $\varphi^\star$. This includes a broad class of dependence structures, such as Gaussian, elliptical, and certain Archimedean copula models, while the condition $\Lambda_i\in\RV(\alpha_i)$ allows for a wide range of light-tailed marginals. 

\begin{assumption}[\textbf{Asymptotic homogeneity of loss}]\label{assume:loss}
The loss function $L(\cdot)$ satisfies
\begin{equation}\label{eqn:loss}
n^{-1}L(n\zz_n)\to L^\star(\zz)
\end{equation}
for every sequence $\zz_n\to\zz$, where $L^\star$ is a non-degenerate limiting function in the sense that $\{\zz:L^\star(\zz)>0\}\neq \varnothing$.
\end{assumption}
{Assumption~\ref{assume:loss} only requires that the loss admit a
first-order homogeneous approximation at high stress levels; this limiting
shape determines which configurations dominate the stressed regime. Nonlinear losses such as $L(\zz)=\|\zz\|_2$ satisfy the condition directly. More generally, if $L_0$ has tail order $r>0$, meaning that $n^{-r}L_0(n\zz_n)\to L_0^\star(\zz)$, then the transformed loss $L=L_0^{1/r}$ satisfies $n^{-1}L(n\zz_n)\to \bigl(L_0^\star(\zz)\bigr)^{1/r}$. Hence losses with general tail order can be put into the normalization used in Assumption~\ref{assume:loss}. The next two examples illustrate this flexibility for thresholded linear and network-clearing losses.} 
\begin{example}[\textbf{Bipartite Reinsurance Network}]\label{eg:reinsurance}\em
Consider the bipartite agent-object market model (see  \cite{kley2016risk}), in which agent exposures are linear in object losses. Let $\XX=(\xi_1,\dots,\xi_d)\in\mathbb R_+^d$ denote object losses, and let $A_{ij}\ge 0$ denote the exposure of agent $i$ to object $j$, with $A_{ij}=0$ whenever agent $i$ is not exposed to object $j$. Writing $A\XX= ((A\XX)_1,\ldots, (A\XX)_q) =  (\mv a_1^\intercal \XX, \ldots, \mv a_q^\intercal \XX)$ for the vector of agent losses, and letting $c_i>0$ denote the capital of agent $i$, define the aggregate market distress by
\begin{equation}\label{eqn:loss_reinsurance}
L(\XX):=\sum_{i=1}^q\bigl(\mv a^\intercal_i \XX -c_i\bigr)_+.
\end{equation}
\begin{lemma}\label{lem:verify_reinsurance}
The loss function in \eqref{eqn:loss_reinsurance} satisfies Assumption~\ref{assume:loss}. In particular, $L^\star(\zz)=\sum_{i=1}^q \mv a_i^\intercal \zz$. 
\end{lemma}
\end{example}
The next example shows that the asymptotic homogeneity condition also holds in a nontrivial network model with contagion effects. 
\begin{example}[\textbf{State-Dependent Clearing Network}]\label{eg:clearing}\em
Consider a banking network with $m$ firms and exogenous loss drivers $\XX=(\xi_1,\dots,\xi_d)\in\mathbb R_+^d$. Firm $j$ incurs gross loss $C_j(\XX):=\mv a_j^\intercal \XX$ for $\mv a_j\in\mathbb R_+^d$,
and has state-dependent nominal liabilities
\[
\bar p_j(\XX):=\bigl(C_j(\XX)-r_j\bigr)_+, \qquad r_j\ge 0.
\]
Thus, liabilities activate only once the loss of firm $j$ exceeds its retention level $r_j$.

The network structure is described by fixed relative liability weights $\lambda_{ji}\ge 0$, where $\lambda_{ji}$ denotes the fraction of firm $j$'s total liabilities owed to firm $i$. Assume that $\sum_{i=1}^m \lambda_{ji}=1$ for each $j=1,\dots,m$. The nominal liability of firm $j$ to firm $i$ is 
\[
R_{ij}(\XX):=\lambda_{ji}\,\bar p_j(\XX).
\]
Let $x_j(\XX)\ge 0$ denote the external resources of firm $j$ in state $\XX$, that is, funds from outside the interbank network available to meet its liabilities, such as recoveries from external assets or payments from non-network counter-parties.  The clearing vector $p(\XX)$ solves the system (see \cite{eisenberg2001systemic}) 
\[
p_j(\XX)=\min\Big\{\bar p_j(\XX),\ x_j(\XX)+\sum_{k=1}^m \lambda_{kj}\,p_k(\XX)\Big\},
\qquad j=1,\dots,m.
\]
Since payments are made pro rata, the realized payment from firm $j$ to firm $i$ is $\lambda_{ji}p_j(\XX)$. The unpaid liability is therefore
\[
U_{ij}(\XX):=R_{ij}(\XX)-\lambda_{ji}p_j(\XX)
=\lambda_{ji}\bigl(\bar p_j(\XX)-p_j(\XX)\bigr).
\]
We take the aggregate loss to be total unpaid obligations:
\[
L(\XX):=\sum_{i=1}^m\sum_{j=1}^m U_{ij}(\XX)
=\sum_{j=1}^m \bigl(\bar p_j(\XX)-p_j(\XX)\bigr),
\]
where the second equality uses $\sum_{i=1}^m \lambda_{ji}=1$. If external resources admit a corresponding first-order limit, then the clearing system itself admits a homogeneous asymptotic counterpart. 

\begin{lemma}\label{lem:eisenberg_clearing}
Assume that for each $j=1,\dots,m$, $t^{-1}x_j(t\zz_t)\to x_j^\star(\zz)$ for every sequence $\zz_t\to\zz$, and define $\bar p_j^\star(\zz):=\mv a_j^\intercal \zz$. Assume further that, for each $\zz$, the limiting clearing system
\[
p_j^\star(\zz)=\min\Big\{\bar p_j^\star(\zz),\ x_j^\star(\zz)+\sum_{k=1}^m \lambda_{kj}\,p_k^\star(\zz)\Big\},
\qquad j=1,\dots,m,
\]
admits a unique clearing vector $p^\star(\zz)$. Then the loss function $L(\cdot)$ satisfies Assumption~\ref{assume:loss} with $L^\star(\zz)=\sum_{j=1}^m \bigl(\bar p_j^\star(\zz)-p_j^\star(\zz)\bigr)$.
\end{lemma}
\end{example}
\section{STRESS TESTING VIA LARGE DEVIATIONS} \label{sec:LDP_Stress}
\subsection{A Large-Deviations Benchmark for Stress Scenarios}
{
We begin by identifying the tail states through which the stress event $\{L(\XX)>u\}$ is most likely to occur at large $u$. The large-deviations characterization below expresses this through a rate minimization problem. The minimizers form the dominant set: among all stress-producing points, they have the smallest large-deviations rate and hence the largest probability on the exponential scale. We say that scenarios are plausible when, after tail standardization, they lie near this set. The construction is endogenous because the set is determined by $(\varphi^\star,L^\star)$, rather than by an externally imposed shock.}

Under Assumptions~\ref{assume:ldp}-\ref{assume:loss}, Theorem~2 of
\cite{deo2025achieving} implies that as $u\to\infty$
\begin{equation}\label{eqn:LDP_asymp}
  \mathbb P\!\big(L(\XX)>u\big)
=
\exp\!\big(-t(u)\,(\kappa^\star+o(1))\big), \quad \text{ where } \quad  \kappa^\star
:=
\min\left\{
\varphi^\star(\zz)\;:\;
L^\star\!\big(\qq^\star \zz^{1/\mv\alpha}\big)\ge 1
\right\}  
\end{equation}
and $t(u):=\min_{1\le i\le d}\Lambda_i(u)$.
Thus, $\kappa^\star$ governs the asymptotic exponential decay of the stress event
$\{L(\XX)>u\}$. To characterize the configurations that achieve this decay rate, define the feasible set and its minimizers 
\[
\Gamma
:=
\left\{
\zz\;:\;
L^\star\!\big(\qq^\star \zz^{1/\mv\alpha}\big)\ge 1
\right\} \text{ and }
\quad \mathcal S
:=
\arg\min\bigl\{\varphi^\star(\zz):\zz\in\Gamma\bigr\}
\]
respectively. 
{ Elements of $\mathcal S$ describe the dominant tail configurations at the
large-deviations scale, and therefore form the benchmark set around which
plausible stress scenarios concentrate}. Under strict convexity of $\varphi^\star$, this minimizer
is unique, though the analysis requires only that $\mathcal S$ be nonempty and closed. To formalize concentration near $\mathcal S$, { for $\delta>0$ let
\[
\mathcal M_{\delta,u}
:=
\qq\bigl(t(u)\,\mathcal S_\delta^c\bigr) = \{\qq(t(u)\yy): \yy\in \mathcal S_\delta^c \}
\]
} where $\mathcal S_\delta $ is the set of points at a distance of at most $\delta$ from $\mathcal S$. { Thus, $\mathcal M_{\delta,u}$ collects scenarios whose tail-standardized
coordinates lie outside the $\delta$-neighborhood of the dominant set
$\mathcal S$.}
\begin{proposition}[\textbf{Stress scenarios via large deviations}]
\label{prop:ldp_stress_perspective}
Let Assumptions~\ref{assume:ldp}-\ref{assume:loss} hold. Then for any $\delta>0$,
\[
\limsup_{u\to\infty}
[t(u)]^{-1}
\log
\mathbb P\!\left(
\XX\in\mathcal M_{\delta,u}\,\middle|\,L(\XX)>u
\right)
<0.
\]
\end{proposition}
\noindent 
{Proposition~\ref{prop:ldp_stress_perspective} provides the large-deviations justification for plausibility. Since $\XX\in\mathcal M_{\delta,u}$ is equivalent to $\YY_u:=[t(u)]^{-1}\mv\Lambda(\XX)\notin\mathcal S_\delta$, the proposition implies that, conditional on the stress event $\{L(\XX)>u\}$, the tail-standardized vector $\YY_u$ concentrates near the dominant set $\mathcal S$. In particular, for any fixed $\delta>0$, $P\left(\YY_u\notin\mathcal S_\delta \mid L(\XX)>u\right)$ is exponentially small at speed $t(u)$. This motivates the interpretation that a scenario $\zz$ is plausible at level $u$ when its standardized tail coordinates $[t(u)]^{-1}\Lambda(\zz)$ lie near $\mathcal S$. Scenarios away from this set carry negligible mass in the conditional stress regime.}
\begin{definition}[\textbf{Asymptotically representative stress distribution}]
\label{def:rep_stress_dist}
A random vector $\tilde\XX$ is an \emph{asymptotically representative stress distribution} for
$L(\cdot)$ if $\mathbb P(L(\tilde\XX)>u)>0$ for all sufficiently large $u$ and, for every
$\delta>0$,
\[
\limsup_{u\to\infty}
[t(u)]^{-1}
\log
\mathbb P\!\left(
\tilde\XX\in\mathcal M_{\delta,u}\,\middle|\,L(\tilde\XX)>u
\right)
<0.
\]
\end{definition}
Definition~\ref{def:rep_stress_dist} requires that, under stress, $\tilde\XX$ reproduce the same
tail geometry as the baseline model after marginal standardization. Because $\mathcal M_{\delta,u}$ is constructed from the baseline transforms $\mv\Lambda$ and $\qq$, the criterion compares stress distributions using the baseline tail scaling and the baseline large-deviations geometry. Both features are essential: even a distribution with correctly specified
marginals will fail to be representative if its exceedances are routed through the wrong
large-deviations configurations.
\subsection{A Self-Structuring Transform for Representative Stress Generation}
We now construct a transformed distribution designed to meet this benchmark while making target exceedances more frequent. For $s>1$, define the self-structuring map $\mv T_s:\mathbb R_+^d\to\mathbb R_+^d$ componentwise by
\begin{equation}\label{eqn:T_s}
T_{s,j}(\zz)
:=
z_j\,s^{\kappa_j(\zz)},
\qquad
\kappa_j(\zz)
:=
\frac{\log(1+z_j)}{\max_{1\le i\le d}\log(1+z_i)},
\qquad j=1,\dots,d.
\end{equation}

{ Self-structuring maps were introduced in \cite{deo2025achieving} for loss-agnostic importance sampling, where they define a proposal law for estimating $P(L(\XX)>u)$. Here, we use the same tail-extrapolation mechanism for a different purpose: transforming observed samples into an empirical distribution of representative scenarios for the conditional stress regime. The normalization in $\kappa_j(\zz)$ captures relative tail weights, allowing
$\mv T_s$ to make stress scenarios more frequent in the transformed sample
while preserving the large-deviations geometry of the most likely stress
configurations.} To analyze this effect,
 define
\[
\YY_u:=[t(u)]^{-1}\mv\Lambda(\XX),
\qquad
\YY_{s,u}:=[t(u)]^{-1}\mv\Lambda\!\bigl(\mv T_s(\qq(t(u)\YY_u))\bigr).
\]
\begin{proposition}\label{prop:ldp_transform}
Let $\YY_{s,u}$ be defined as above. Then $\YY_{s,u}\in \LDP(t,\varphi_s^\star)$, where
\[
\varphi_s^\star(\zz)=s^{-\alpha_\star}\varphi^\star(\zz).
\]
\end{proposition}
Proposition~\ref{prop:ldp_transform} shows that the self-structuring transform preserves the
baseline tail structure up to a positive scalar rescaling of the rate function. Because
$\varphi_s^\star$ is a positive multiple of $\varphi^\star$, the minimizer set remains unchanged.
Thus, the transformed distribution
$ \tilde\XX_s:=\mv T_s(\XX)$
preserves the same most-likely stress configurations as the baseline model while making tail
exceedances occur more frequently.

\begin{theorem}\label{thm:representative_transform}
Let Assumptions~\ref{assume:ldp}-\ref{assume:loss} hold. Then for any $\delta>0$,
\[
\limsup_{u\to\infty}
[t(u)]^{-1}
\log
\mathbb P\!\left(
\tilde\XX_s\in\mathcal M_{\delta,u}\,\middle|\,L(\tilde\XX_s)>u
\right)
<0,
\]
and consequently $\tilde\XX_s$ is an asymptotically representative stress distribution in the
sense of Definition~\ref{def:rep_stress_dist}.
\end{theorem}

Theorem~\ref{thm:representative_transform} gives the main justification for self-structuring
stress generation: the transformed law makes exceedances less rare without changing the most-likely
stress configurations. 

\subsection{An Empirical Plug-In Generator}
{Given data $\{\XX^1,\ldots,\XX^n\}$, Algorithm~\ref{algo:ssgen} approximates the benchmark conditional stress distribution $\mathrm{Law}(\XX\mid L(\XX)>u)$ from Section~\ref{sec:setup}. Theorem~\ref{thm:representative_transform} provides the large-deviations
justification: the transformed law $\widetilde\XX_s=\mv T_s(\XX)$ preserves the
most-likely configurations while making threshold exceedances more frequent.}
\begin{algorithm}[h]
  \caption{Self-Structuring Stress Scenario Generator}
  \label{algo:ssgen}
  \textbf{Input: } Data samples $\XX_1,\ldots,\XX_n\in\mathbb R_+^d$, loss function $L(\cdot)$, target stress threshold $u$, stretch parameter $s>1$

\noindent\textbf{I) Transformation Step:} For $i=1,\ldots,n$, set $\tilde\XX_i=\mv T_s(\XX_i)$.

\noindent\textbf{II) Construct Stress Samples:} Compute the transformed exceedance index set $\mathcal I_{u,s}
:=
\left\{
i\in[n]:L(\tilde\XX_i)>u
\right\}$,
and form the transformed exceedance pool
\[
\mathcal G_{u,s}
:=
\left\{
\tilde\XX_i:i\in\mathcal I_{u,s}
\right\}.
\]
\vspace{-8
pt}
\noindent\textbf{III) Output Stress Distribution:} Construct the empirical stress distribution
\[
\hat P^{\mathrm{stress}}_{u,s}
:=
{|\mathcal I_{u,s}|}^{-1}
\sum_{i\in\mathcal I_{u,s}}
\delta_{\tilde\XX_i},
\]

\noindent\textbf{Output: } Stress scenarios $\mathcal G_{u,s}$ and empirical stress law $\hat P^{\mathrm{stress}}_{u,s}$.
\end{algorithm}

{The algorithm implements this idea by applying $\mv T_s$ to the baseline samples, retaining transformed exceedances, and forming their empirical distribution. Thus, its output is a finite-sample proxy, not an exact conditional sampler. Apart from evaluating $L$, the procedure requires only one componentwise transformation and one exceedance check per input sample, making it computationally light.}

\begin{illustration}\em
We illustrate Algorithm~\ref{algo:ssgen} on a reinsurance network with $d=2$ objects and $q=3$ agents. Figure~\ref{fig:GC_true} shows that a Gaussian copula, even with matched marginals and covariance, misses the true tail dependence structure. A naive empirical sampler is omitted because realistic baseline samples typically yield no exceedances at these quantiles. Figure~\ref{fig:SS_true} shows that the empirical stress law $\hat P_{u,s}^{\mathrm{stress}}$ from Algorithm~\ref{algo:ssgen} more closely matches the true stressed distribution.
\begin{figure}[htbp]
    \centering
    \begin{subfigure}[t]{0.3\textwidth}
        \centering
\includegraphics[width=\linewidth]{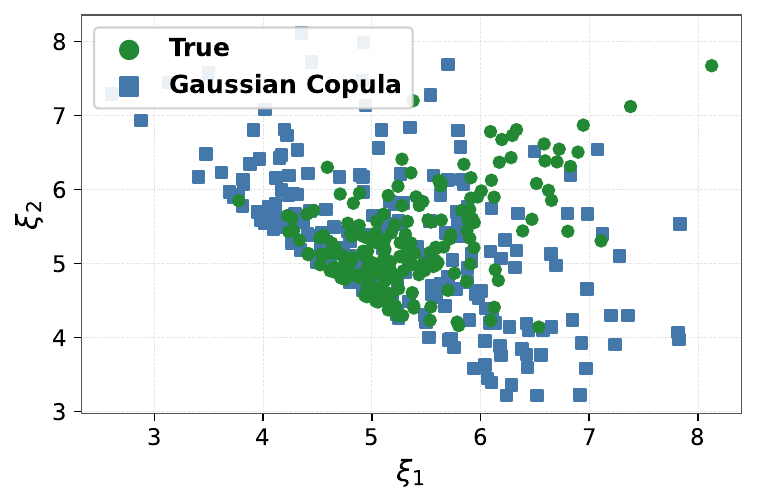}
        \caption{Gaussian copula generator}
        \label{fig:GC_true}
    \end{subfigure}
    \hspace{0.015\textwidth}
    \begin{subfigure}[t]{0.3\textwidth}
        \centering
        \includegraphics[width=\linewidth]{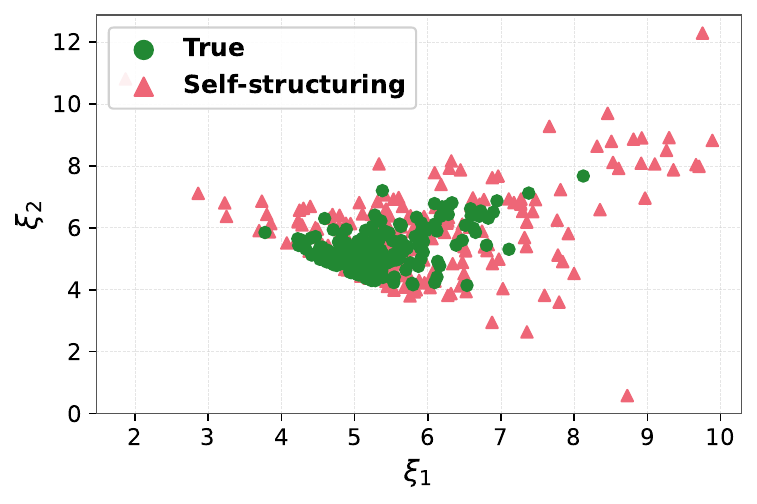}
        \caption{Samples from Algorithm~\ref{algo:ssgen}}
        \label{fig:SS_true}
    \end{subfigure}
   \caption{Conditional distribution of risk factors given that the loss exceeds a high threshold. Algorithm~\ref{algo:ssgen} produces samples that more closely match the true stressed distribution than the Gaussian-copula benchmark, while remaining implementable at stress levels where a naive empirical exceedance resampler would produce no samples.}
    \label{fig:comparison}
\end{figure}
\end{illustration}

\noindent \textbf{Sample Complexity Reduction using Self-Structuring Transformations: }
A representative stress distribution is useful only if it can be learned from finite data. For
practically relevant thresholds $u$, however, the exceedance probability $p_u:=\mathbb P(L(\XX)>u)$
is very small. Naive sampling from the baseline model yields only
$np_u$ exceedances on average. Thus, observing even a modest number of stress scenarios requires a
sample size of order $p_u^{-1}$. The next proposition shows that the self-structuring transform
$\mv T_s$ mitigates this difficulty by making exceedances occur at a strictly larger rate on the
logarithmic scale.

\begin{proposition}\label{prop:ease_of_learning}
Let $\tilde\XX_s=\mv T_s(\XX)$ and let $\XX$ satisfy Assumptions~\ref{assume:ldp}--\ref{assume:loss}. Then
\[
\frac{\log \mathbb P(L(\XX)>u)}{\log \mathbb P(L(\tilde\XX_s)>u)}
\longrightarrow s^{\alpha_\star},
\qquad u\to\infty.
\]
\end{proposition}

Proposition~\ref{prop:ease_of_learning} implies that
$\mathbb P\bigl(L(\tilde\XX_s)>u\bigr)
\approx
\mathbb P\bigl(L(\XX)>u\bigr)^{s^{-\alpha_\star}}
=
p_u^{\,s^{-\alpha_\star}}$,
so exceedances are exponentially less rare under $\tilde\XX_s$ than under the baseline model.
Consequently, the sample size required to observe a fixed expected number of exceedances is reduced
from order $p_u^{-1}$ to order $p_u^{-s^{-\alpha_\star}}$. Equivalently, from $n$ baseline samples,
the expected number of transformed exceedances at level $u$ is of order
$n\,p_u^{\,s^{-\alpha_\star}}.
$
Thus one can still expect to observe a nondegenerate number of transformed exceedances whenever
$n\,p_u^{\,s^{-\alpha_\star}}\asymp 1$,
that is, for target events with baseline probability
$p_u \asymp n^{-s^{\alpha_\star}}.
$
In particular, writing $\gamma(s):=s^{\alpha_\star}>1$, one can expect to generate a nondegenerate number of representative transformed exceedances from $n$ baseline samples for target events with probabilities of order $n^{-\gamma(s)}$, including regimes in which the raw dataset contains few or no direct exceedances. The sensitivity
of the procedure to the choice of $s$ is briefly examined in Section~\ref{sec:numericals}.
\section{NUMERICAL EXPERIMENTS} \label{sec:numericals}
We evaluate the proposed stress generator along two dimensions: \textit{representativeness}, namely recovery of the baseline conditional stress law, and \textit{data efficiency}, namely the ability to generate informative stress scenarios when target-level exceedances are rare or absent. The baseline risk factors $\XX$ are generated from a $t$-copula 
with Weibull marginals with cdf $F_i(x) = 1-\exp(-x^{\alpha_i}
)$, so as to capture both extremal-tailed dependence and 
flexible marginal tail decay. For each threshold $u$, we construct a high-fidelity reference sample from $P_u:=\mathcal L(\XX\mid L(\XX)>u)$
using a Monte Carlo sample of size $10^6$ from the baseline model.

\subsection{Experiment 1: Efficient Learning Beyond Observed Data}

We first test whether the proposed generator can recover the stressed conditional law from finite baseline data, especially in regimes where direct exceedances are sparse or unavailable. Let $u$ be a high quantile of the baseline loss $L(\XX)$, and consider the two thresholds $u=q_{0.99}$ and $u=q_{0.999}$. For each input sample size $n\in\{250,500,750,1000\}$, we generate an independent baseline sample of size $n$ and construct three stress generators: (i) the proposed extrapolation-based generator, (ii) a Gaussian-copula generator with correctly specified marginals, and (iii) a naive empirical baseline that resamples observed exceedances when available. Each generator is then used to generate a stress sample targeted at the same threshold $u$.  Performance is measured by the Wasserstein-1 distance between the generated and reference exceedance samples, normalized by $u$. Repeating the experiment over $K=1000$ independent samples yields, for each $n$, a median normalized Wasserstein error together with pointwise $95\%$ empirical bands.

We begin with the reinsurance-network loss from Example~\ref{eg:reinsurance} in a system with $d=10$ risky objects and $q=30$ agents. Figure~\ref{fig:reinsurance_all_three} reports the resulting normalized Wasserstein errors. At the more extreme threshold $u=q_{0.999}$, the Gaussian-copula benchmark fails completely: even after generating $10^5$ samples, it produces no exceedances above the target level. By contrast, the extrapolation-based generator continues to recover the stressed law accurately from finite training data. For the same threshold, we also examine sensitivity to the stretch parameter $s$ by fixing $n=1000$, varying $s$ over the indicated range, and again reporting the median normalized Wasserstein error together with pointwise $95\%$ empirical bands. 
{
We choose $s$ using the following heuristic: for a candidate value of
$s$, let $N_s(u):=\sum_{i=1}^n \mathbf 1\{L(T^s(\XX_i))>u\}$
denote the number of transformed exceedances at level $u$. In each experiment,
we select the smallest $s$ for which 
$N_s(u)$ is roughly $50$. Since $N_s(u)$ is monotone in $s$ under monotone losses, the rule can be implemented by a simple grid search over $s$. This choice is fixed before evaluating the
reported error. Figure~\ref{fig:reins_vary_s} nevertheless shows that it lies
close to the interior minimizer, suggesting robustness to moderate
misspecification of $s$.} 
\begin{figure}[htbp]
    \centering
    \begin{subfigure}[t]{0.32\textwidth}
        \centering
        \includegraphics[width=\linewidth]{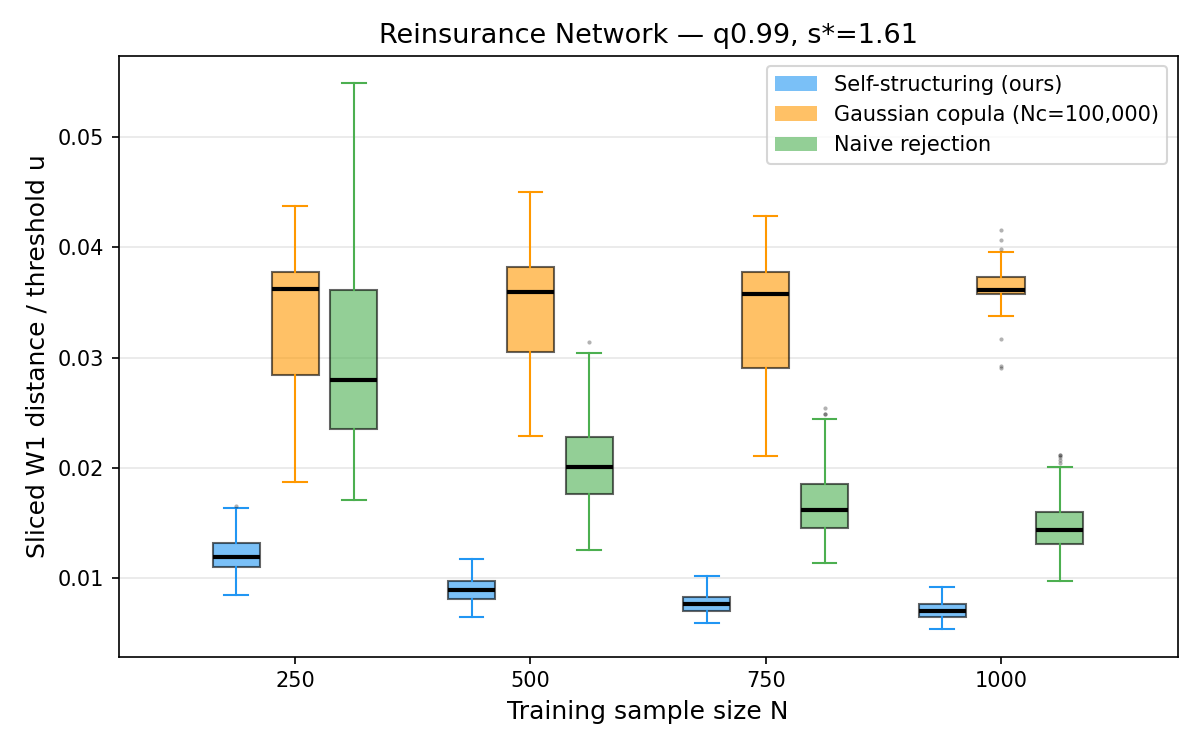}
        \caption{$q=0.99$, $s=1.61$}
        \label{fig:reins_q099}
    \end{subfigure}\hfill
    \begin{subfigure}[t]{0.32\textwidth}
        \centering
        \includegraphics[width=\linewidth]{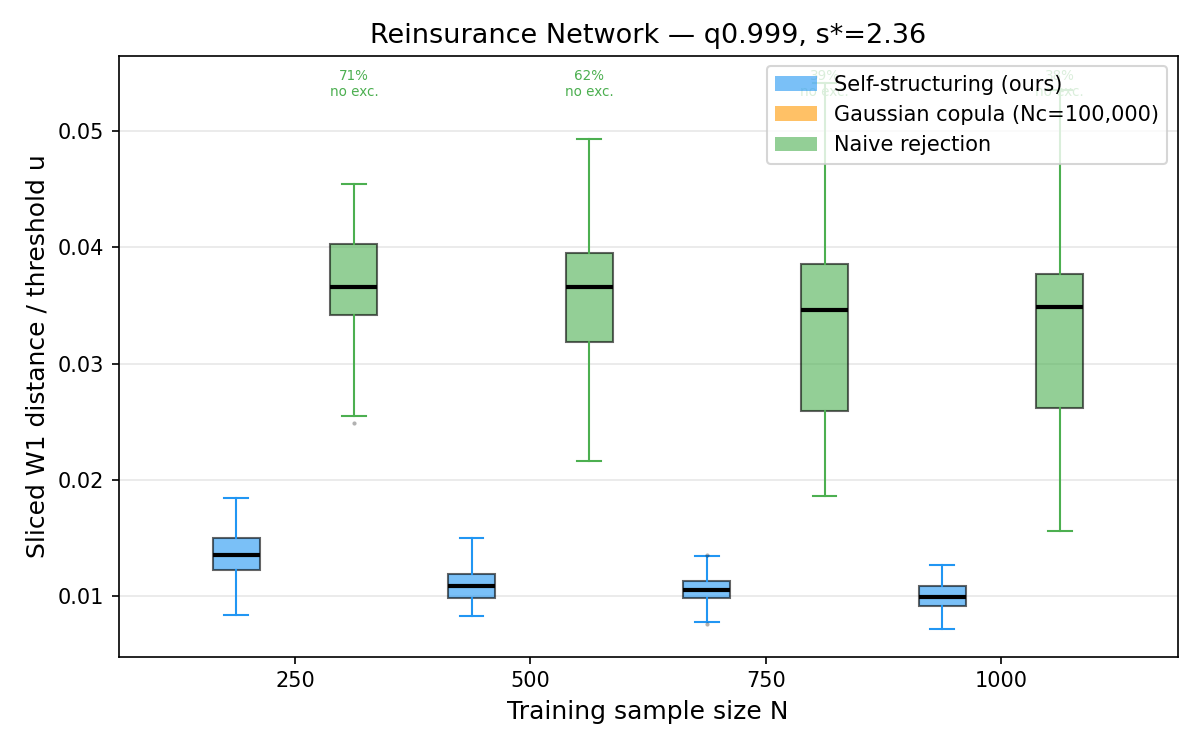}
        \caption{$q=0.999$, $s=2.36$}
        \label{fig:reins_q0999}
    \end{subfigure}\hfill
    \begin{subfigure}[t]{0.32\textwidth}
        \centering
        \includegraphics[width=\linewidth]{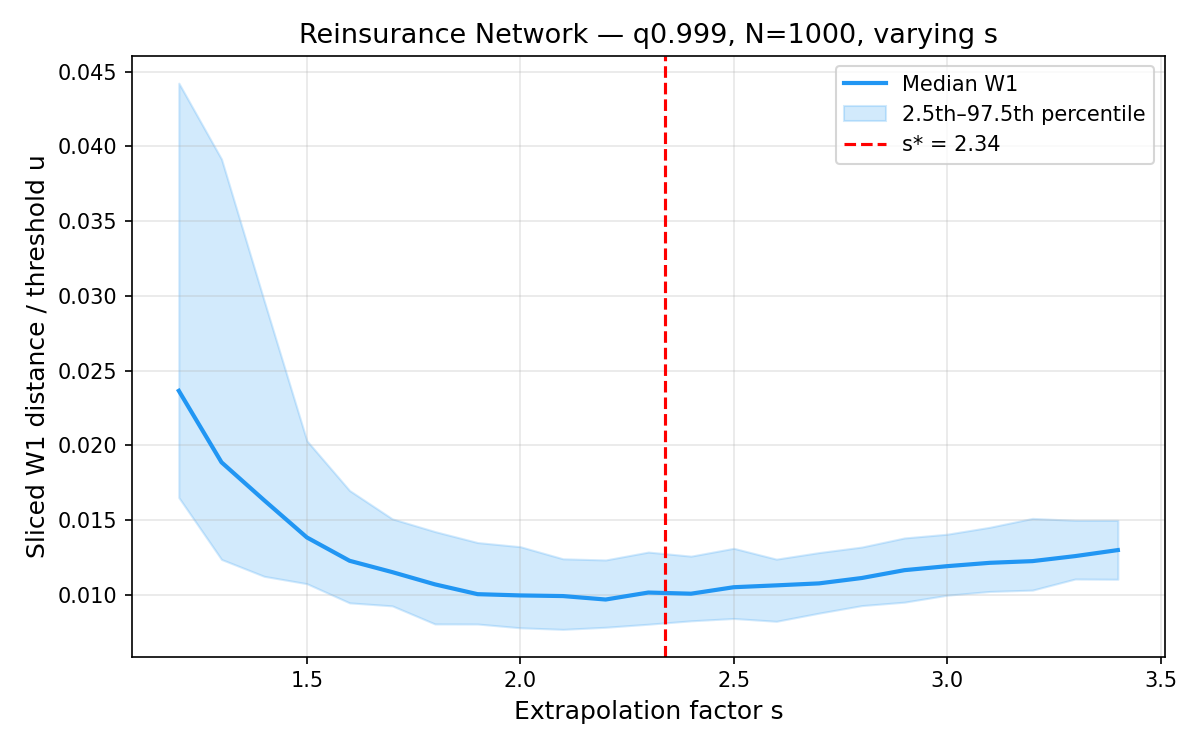}
        \caption{$q=0.999$, varying $s$}
        \label{fig:reins_vary_s}
    \end{subfigure}
   \caption{Normalized Wasserstein-1 error for the reinsurance-network loss at $u=q_{0.99}$ 
(left) and $u=q_{0.999}$ (middle), and sensitivity to the stretch parameter $s$ 
at $u=q_{0.999}$ with $n=1000$ (right). Here the marginals are Weibull with $\alpha_i=1.2$}
    \label{fig:reinsurance_all_three}
\end{figure}

We next examine the proposed method for the clearing network from Example~\ref{eg:clearing} in a system with $m=30$ banks and $d=20$ exogenous risk factors. We consider three stylized network topologies: hub-and-spoke, core-periphery, and fully connected. The topology is specified through the support of the relative liability weights $\lambda_{ji}$. In the fully connected case, $\lambda_{ji}>0$ for all $i\neq j$. In the hub-and-spoke case, one hub bank is connected to every other bank, and no peripheral banks are linked among themselves. In the core-periphery case, banks are partitioned into a core and a periphery, with dense links within the core, no links within the periphery, and links between the core and the periphery. This stylized specification is standard in the interbank tiering literature; see \cite{craig2014interbank}.

Figure~\ref{fig:networks} reports the corresponding normalized Wasserstein errors for the three topologies at $u=q_{0.99}$. Across all three network structures, the extrapolation-based generator more accurately reproduces the reference stress law than either the Gaussian-copula benchmark or the naive empirical baseline. The former misses the relevant tail dependence, whereas the latter is limited by the small number of observed exceedances. Taken together with Figure~\ref{fig:reins_vary_s}, these results show that the proposed generator remains accurate across materially different stress models and is relatively insensitive to the choice of stretch parameter $s$.

\begin{figure}[htbp]
    \centering
    \begin{subfigure}[t]{0.32\textwidth}
        \centering
        \includegraphics[width=\linewidth]{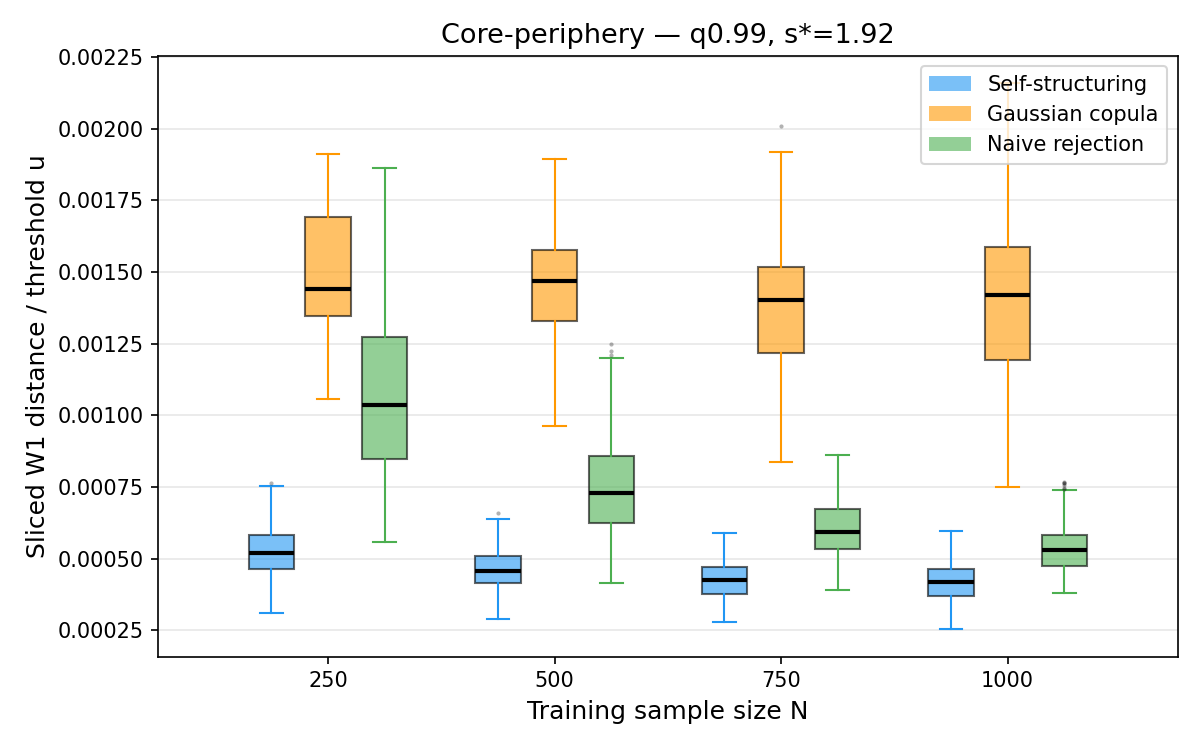}
        \caption{Core--periphery}
        \label{fig:clearing_core_periphery}
    \end{subfigure}\hfill
    \begin{subfigure}[t]{0.32\textwidth}
        \centering
        \includegraphics[width=\linewidth]{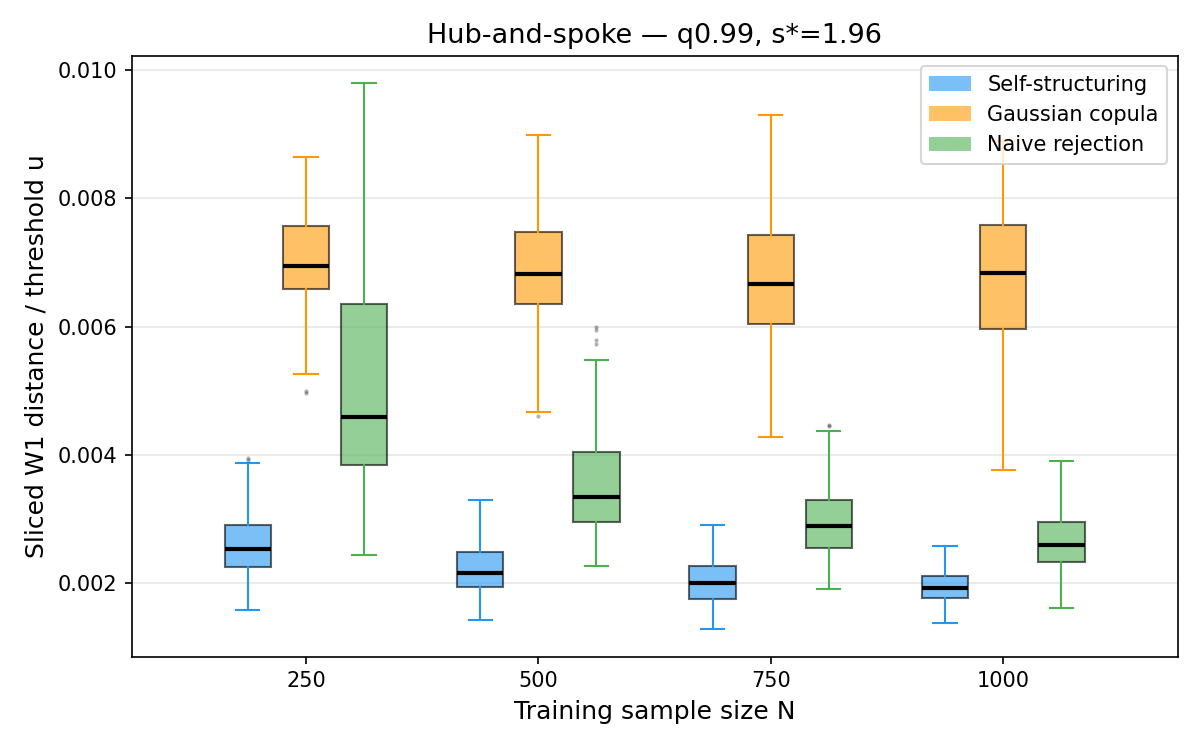}
        \caption{Hub-and-spoke}
        \label{fig:clearing_hub_spoke}
    \end{subfigure}\hfill
    \begin{subfigure}[t]{0.32\textwidth}
        \centering
        \includegraphics[width=\linewidth]{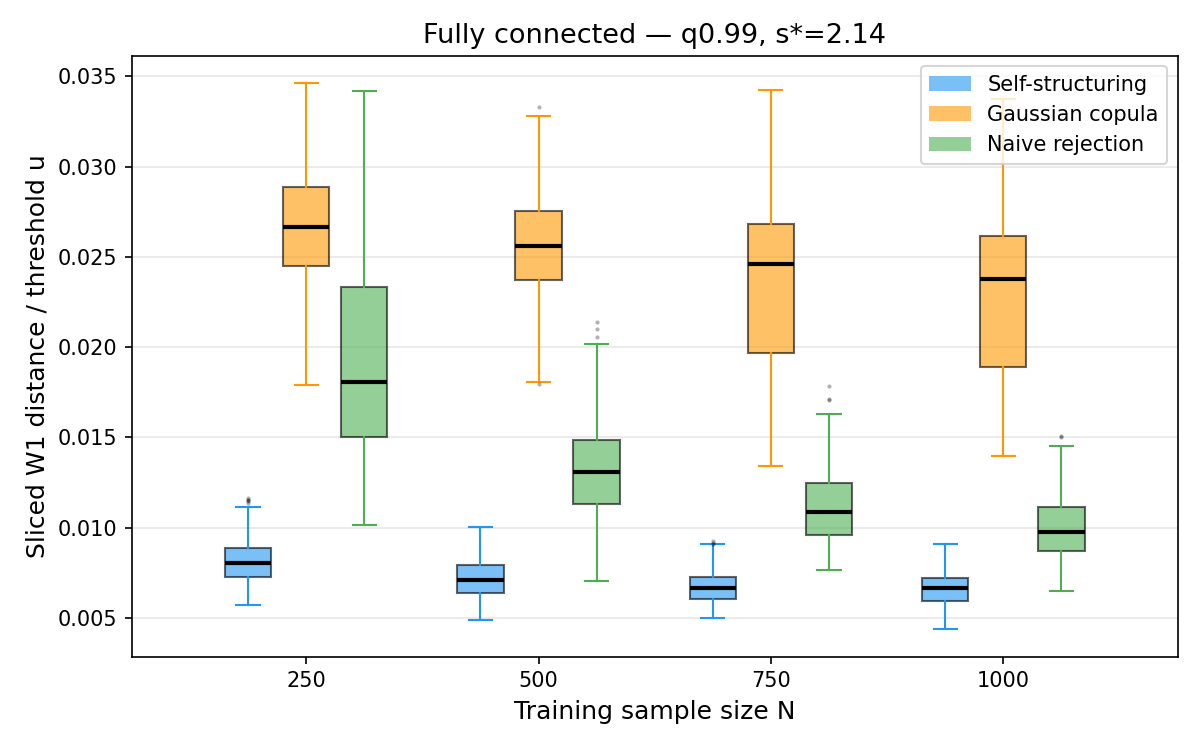}
        \caption{Fully connected}
        \label{fig:clearing_fully_connected}
    \end{subfigure}
   \caption{Normalized Wasserstein-1 error for three stylized 
topologies, all at $u=q_{0.99}$. Box plots summarize median and $95\%$ empirical bands 
across $K=1000$ independent replications for the three generators: self-structuring, 
Gaussian copula, and naive empirical resampler. Here the marginals are Weibull with $\alpha_i = 0.9$.}
    \label{fig:networks}
\end{figure}
\subsection{Experiment 2: Stress Diagnostics Under the Conditional Law}
We next test whether the proposed extrapolation reproduces the internal structure of stress under the conditional law $P_u=\textrm{ Law}(\XX\mid L(\XX)>u)$, beyond the stressed loss law itself.  In both Example~\ref{eg:reinsurance} and Example~\ref{eg:clearing},  the aggregate loss admits a natural institution-level decomposition, 
$L(\XX)=\sum_{i=1}^{n_{\rm inst}} L_i(\XX)$,
where $L_i(\XX)$ denotes the contribution of institution $i$ to total loss and $n_{\rm inst}$ is the number of participating institutions. To distinguish breadth from concentration of distress, we consider the diagnostics
\[
G_{\mathrm{frac}}(\XX):=\frac{1}{n_{\rm inst}}\sum_{i=1}^{n_{\rm inst}} \mathbf 1\{L_i(\XX)>0\},
\qquad
G_{\mathrm{HHI}}(\XX):=\sum_{i=1}^{n_{\rm inst}} \left(\frac{L_i(\XX)}{L(\XX)}\right)^2,
\qquad
G_{\mathrm{loss}}(\XX):=L(\XX).
\]
Here $G_{\mathrm{frac}}(\XX)$ records the fraction of distressed institutions, while $G_{\mathrm{HHI}}(\XX)$ measures the concentration of aggregate loss, ranging from $n_{\rm inst}^{-1}$ under uniform loss sharing to $1$ when a single institution accounts for the entire loss. For the breadth and loss diagnostics, $G_{\mathrm{frac}}$ and $G_{\mathrm{loss}}$, we compare the conditional stressed distributions through their complementary CDFs,
$\bar F_u^G(t):=P(G(\XX)>t\mid L(\XX)>u)$.
For each training sample size $n$, we generate $K=1000$ independent training samples, construct the corresponding stress generator, and estimate the conditional CCDF from the resulting stress sample. We then compare these conditional tail CDFs with the reference conditional CCDF and plot the corresponding pointwise $95\%$ empirical interval as a function of $t$. For the concentration diagnostic $G_{\mathrm{HHI}}$, we instead summarize each generated stress sample by its conditional mean,
\[
\mu_u^{\mathrm{HHI}}
:=
E\!\left[G_{\mathrm{HHI}}(\XX)\mid L(\XX)>u\right],
\]
and examine its distribution across simulation replications. This directly tests whether the generator reproduces the average degree of loss concentration in stressed scenarios.

Since both competing methods produce no exceedances at the level $u=q_{0.999}$, this experiment isolates whether the self-structuring generator preserves these stress diagnostics in a regime where the alternatives cannot generate stress samples at all. Across training sizes, the estimated conditional CCDFs for $G_{\mathrm{frac}}$ and $G_{\mathrm{loss}}$ remain close to their reference curves, with the corresponding bands narrowing as $n$ increases. Likewise, the distribution of the estimated conditional mean of $G_{\mathrm{HHI}}$ contracts toward the reference value as $n$ grows. Taken together, these results show that the proposed generator recovers not only the stressed loss law, but also the breadth, concentration, and severity structure of distress, even at a stress level where the competing methods generate no stress scenarios.
\begin{figure}[htbp]
    \centering
    \begin{subfigure}[t]{0.32\textwidth}
        \centering
        \includegraphics[width=\linewidth]{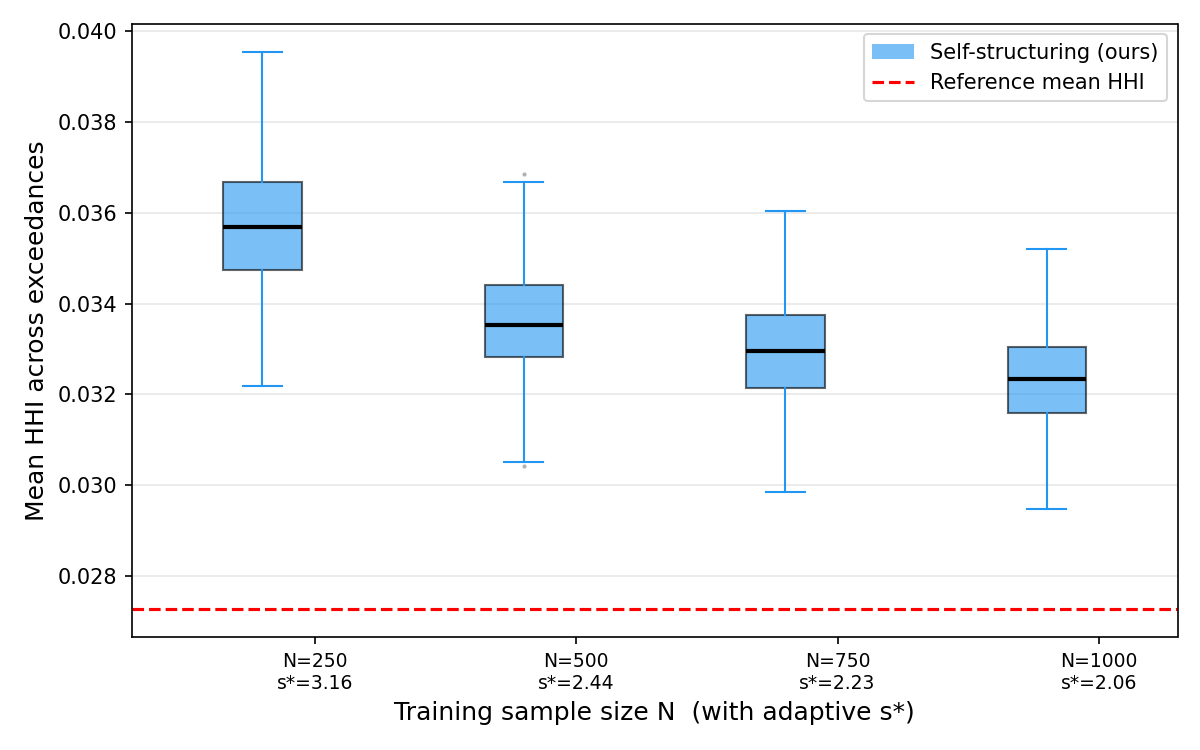}
        \caption{Conditional mean of $G_{\mathrm{HHI}}$}
        \label{fig:HHI_coverage}
    \end{subfigure}\hfill
    \begin{subfigure}[t]{0.32\textwidth}
        \centering
        \includegraphics[width=\linewidth]{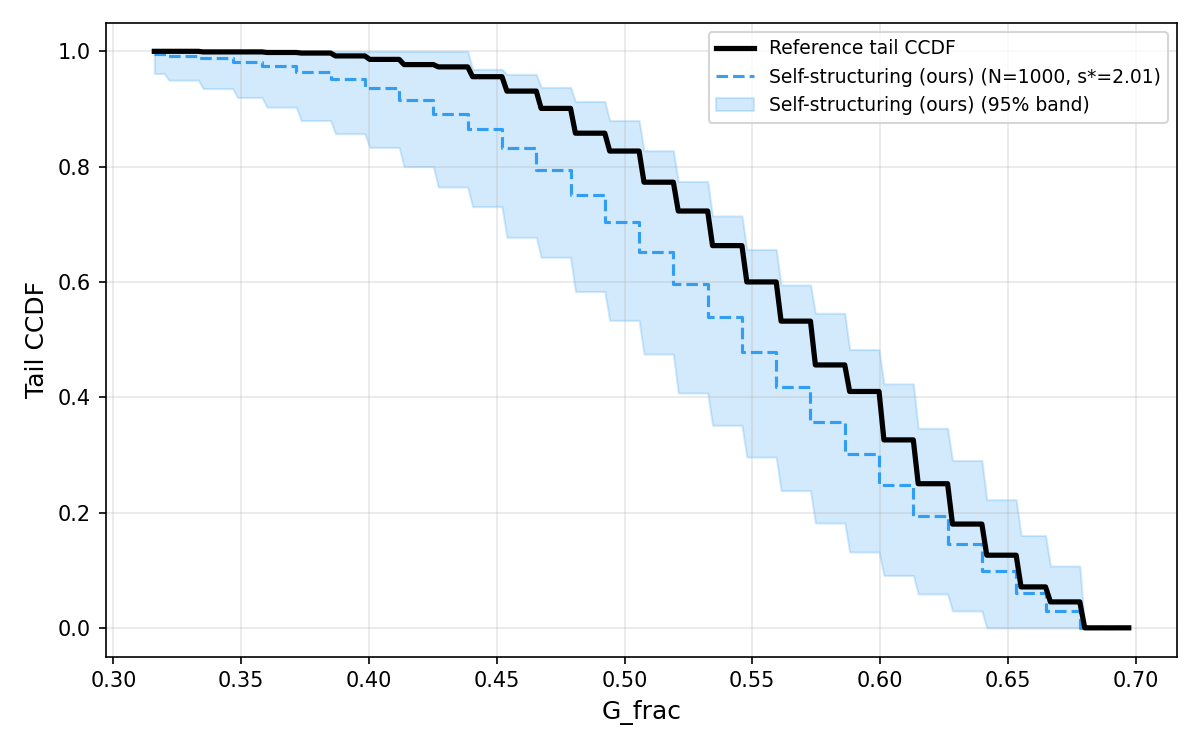}
        \caption{Conditional CCDF of $G_{\mathrm{frac}}$}
        \label{fig:fraction_defaults}
    \end{subfigure}\hfill
    \begin{subfigure}[t]{0.32\textwidth}
        \centering
        \includegraphics[width=\linewidth]{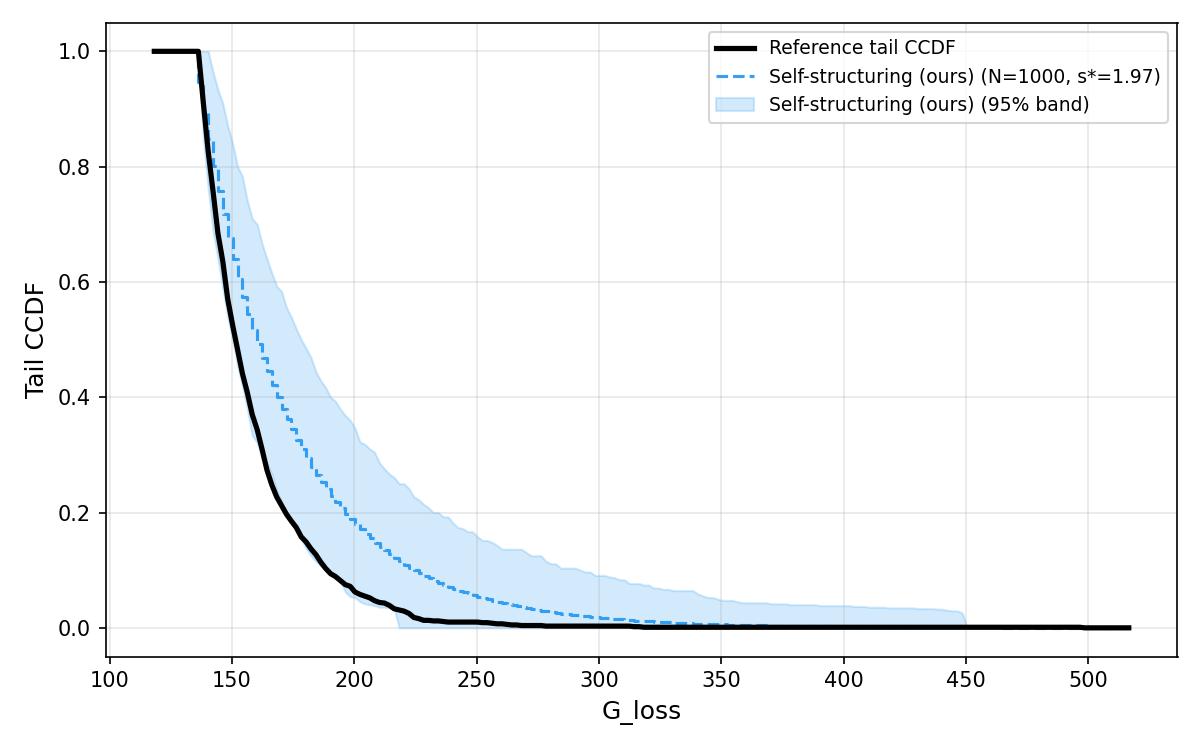}
        \caption{Conditional CCDF of $G_{\mathrm{loss}}$}
        \label{fig:conditonal_loss}
    \end{subfigure}
  \caption{Diagnostic recovery for the reinsurance network at $u=q_{0.999}$. Left: distribution 
of the estimated conditional mean of $G_{\mathrm{HHI}}$ across replications, with the reference 
value marked. Middle and right: pointwise $95\%$ empirical bands for the conditional CCDFs of 
$G_{\mathrm{frac}}$ and $G_{\mathrm{loss}}$, with the reference CCDF overlaid. The value of $s$ is 
chosen so that the self-structuring stress generator produces approximately $50$ samples exceeding $u$.}
    \label{fig:reinsurance_diagnostics}
\end{figure}

\begin{figure}[htbp]
    \centering
    \begin{subfigure}[t]{0.32\textwidth}
        \centering
        \includegraphics[width=\linewidth]{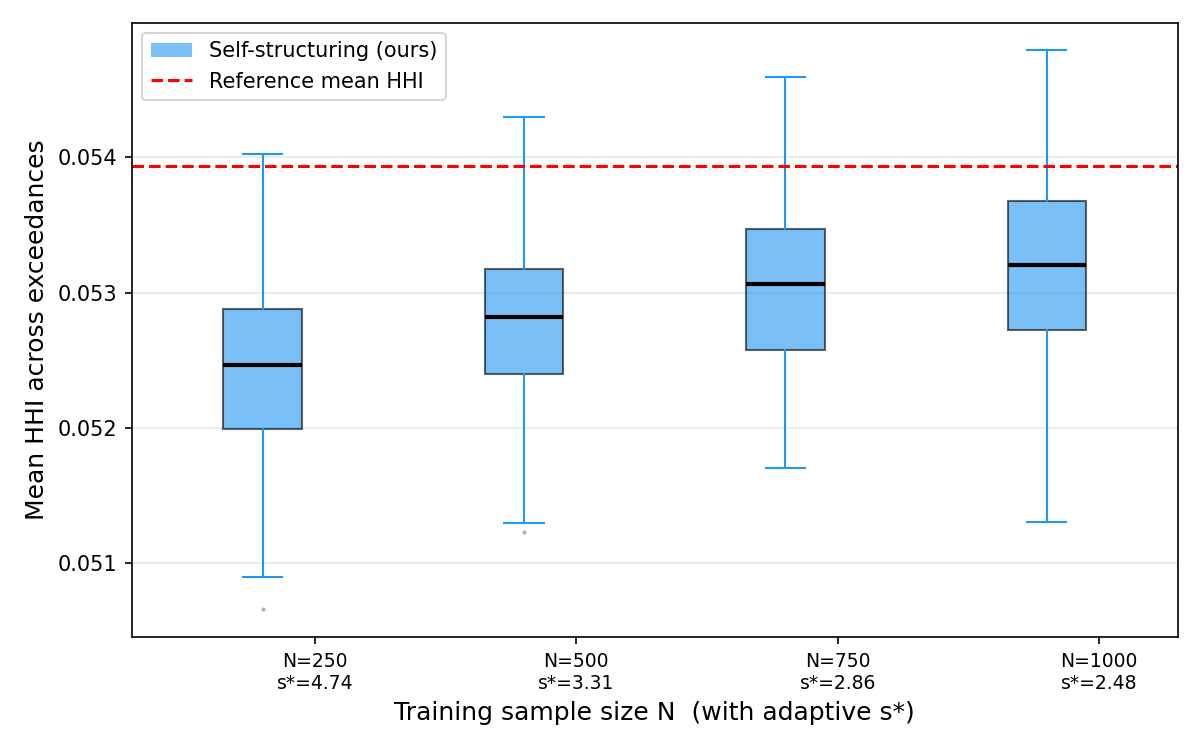}
        \caption{Conditional mean of $G_{\mathrm{HHI}}$}
        \label{fig:HHI_coverage_clearing}
    \end{subfigure}\hfill
    \begin{subfigure}[t]{0.32\textwidth}
        \centering
        \includegraphics[width=\linewidth]{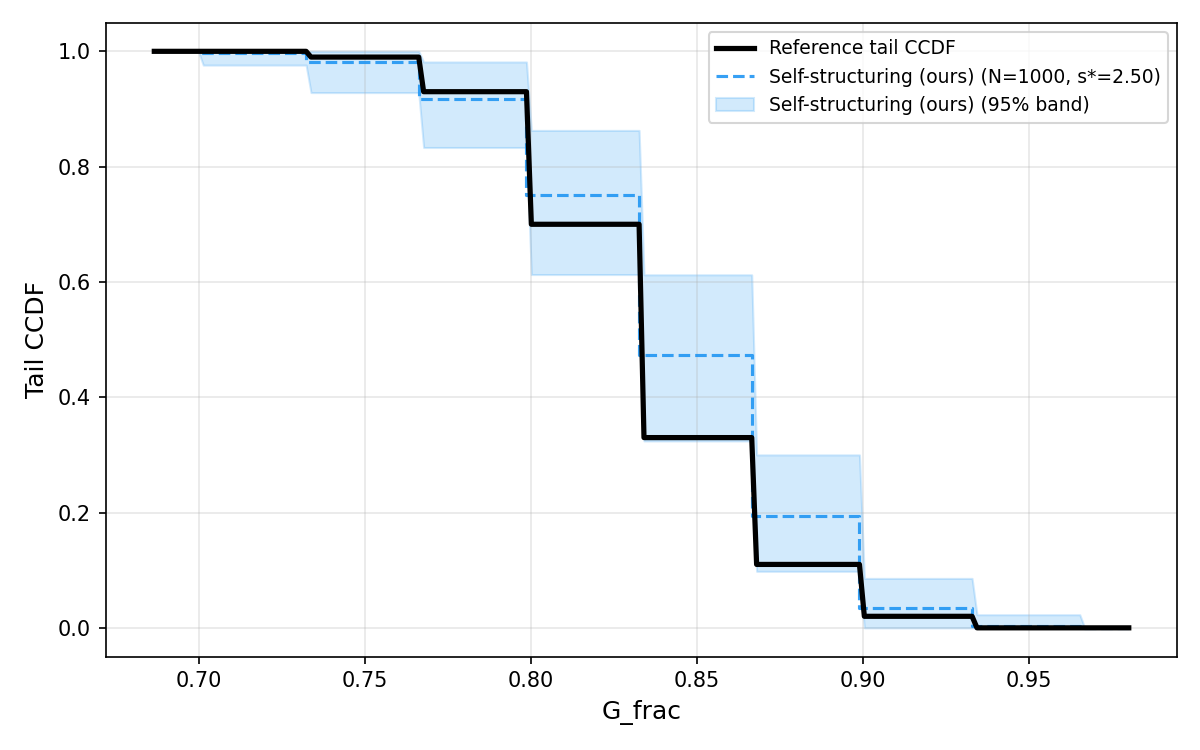}
        \caption{Conditional CCDF of $G_{\mathrm{frac}}$}
        \label{fig:fraction_defaults_clearing}
    \end{subfigure}\hfill
    \begin{subfigure}[t]{0.32\textwidth}
        \centering
        \includegraphics[width=\linewidth]{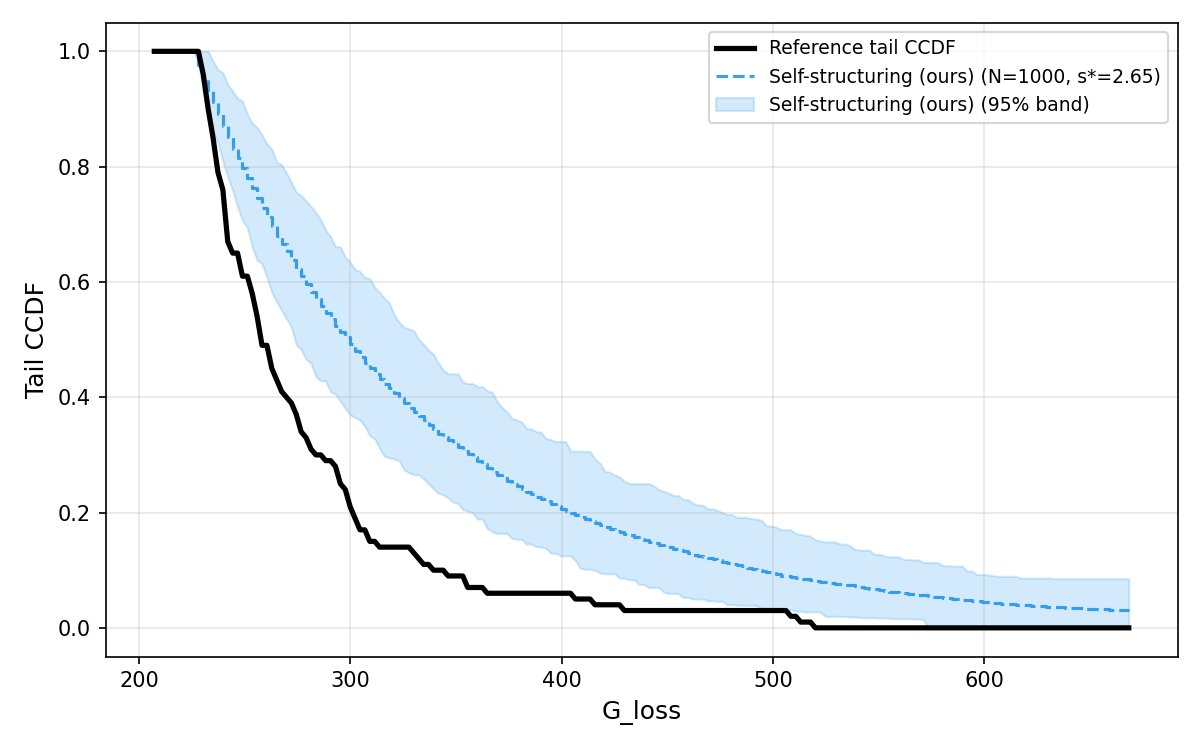}
        \caption{Conditional CCDF of $G_{\mathrm{loss}}$}
        \label{fig:conditonal_loss_clearing}
    \end{subfigure}
  \caption{Diagnostic recovery for the fully connected clearing network at $u=q_{0.999}$. Left: 
distribution of the estimated conditional mean of $G_{\mathrm{HHI}}$ across replications, with 
the reference value marked. Middle and right: pointwise $95\%$ empirical bands for the 
conditional CCDFs of $G_{\mathrm{frac}}$ and $G_{\mathrm{loss}}$, with the reference CCDF 
overlaid.}
    \label{fig:clearing_diagnostics}
\end{figure}
\section{PROOFS}\label{sec:proofs}
In the following proofs, let $L_u(\zz) = u^{-1}L(\qq(t(u)\zz))$. Further, let $\mv{\tilde T}_{s,u}(\zz) = [t(u)]^{-1}\mv{\Lambda}(\mv T_s(\qq(t(u)\zz)))$.

\noindent \textbf{Proof of Lemma~\ref{lem:eisenberg_clearing}}: 
Let $\zz_t\to\zz$. Since $\bar p_j(\xx)
    =
    (\mv a_j^\intercal \xx-r_j)_+$,
we have
\[
    t^{-1}\bar p_j(t\zz_t)
    =
    \bigl(\mv a_j^\intercal \zz_t-t^{-1}r_j\bigr)_+
    \longrightarrow
    (\mv a_j^\intercal \zz)_+
    =
    \mv a_j^\intercal \zz
    =:
    \bar p_j^\star(\zz),
\]
where the last equality uses $\mv a_j\ge \mv 0$ and $\zz\ge \mv 0$. Define
\[
    \bar p_j^{(t)}(\zz_t)
    :=
    t^{-1}\bar p_j(t\zz_t),
    \qquad
    p^{(t)}(\zz_t)
    :=
    t^{-1}p(t\zz_t),
    \qquad
    x_j^{(t)}(\zz_t)
    :=
    t^{-1}x_j(t\zz_t).
\]
Dividing the clearing fixed point by $t$ gives, for each $j=1,\ldots,m$,
\[
    p_j^{(t)}(\zz_t)
    =
    \min\Bigl\{
        \bar p_j^{(t)}(\zz_t),\;
        x_j^{(t)}(\zz_t)
        +
        \sum_{k=1}^m \lambda_{kj}p_k^{(t)}(\zz_t)
    \Bigr\}                                               =
    [\Phi_t(p^{(t)}(\zz_t);\zz_t)]_j,
\]
where the map $\Phi_t(\cdot\,;\zz_t)$ is defined componentwise by
\[
    [\Phi_t(p;\zz_t)]_j
    :=
    \min\Bigl\{
        \bar p_j^{(t)}(\zz_t),\;
        x_j^{(t)}(\zz_t)
        +
        \sum_{k=1}^m \lambda_{kj}p_k
    \Bigr\}.
\]
By hypothesis, $x_j^{(t)}(\zz_t)\to x_j^\star(\zz)$, and from the preceding
calculation, $\bar p_j^{(t)}(\zz_t)\to \bar p_j^\star(\zz)$. Hence
$\Phi_t(\cdot\,;\zz_t)$ converges uniformly on bounded sets to the limiting map
$\Phi^\star(\cdot\,;\zz)$, defined by
\[
    [\Phi^\star(p;\zz)]_j
    :=
    \min\Bigl\{
        \bar p_j^\star(\zz),\;
        x_j^\star(\zz)
        +
        \sum_{k=1}^m \lambda_{kj}p_k
    \Bigr\}.
\]
Since $\mv 0
    \le
    p^{(t)}(\zz_t)
    \le
    \bar p^{(t)}(\zz_t)$, and $\bar p^{(t)}(\zz_t)\to\bar p^\star(\zz)$, the sequence
$\{p^{(t)}(\zz_t)\}_{t\ge 1}$ is eventually bounded. Let
$p^{(t_m)}(\zz_{t_m})$ be any convergent subsequence, and write
$p^{(t_m)}(\zz_{t_m})\to \bar p$. 
Then
\[
\bigl\|\bar p-\Phi^\star(\bar p;\zz)\bigr\|
\le
\bigl\|\bar p-p^{(t_m)}(\zz_{t_m})\bigr\|                                    
+
\bigl\|p^{(t_m)}(\zz_{t_m})
      -\Phi^\star(p^{(t_m)}(\zz_{t_m});\zz)\bigr\|                            
+
\bigl\|\Phi^\star(p^{(t_m)}(\zz_{t_m});\zz)
      -\Phi^\star(\bar p;\zz)\bigr\|.
\]
The first term converges to zero by the definition of the subsequence. The
third term converges to zero because $\Phi^\star(\cdot\,;\zz)$ is continuous.
For the second term, use the fixed point identity
\[
    p^{(t_m)}(\zz_{t_m})
    =
    \Phi_{t_m}(p^{(t_m)}(\zz_{t_m});\zz_{t_m}).
\]
Since the subsequence remains in a bounded set and
$\Phi_t(\cdot\,;\zz_t)\to\Phi^\star(\cdot\,;\zz)$ uniformly on bounded sets, we
obtain
\[
    \bigl\|p^{(t_m)}(\zz_{t_m})
      -\Phi^\star(p^{(t_m)}(\zz_{t_m});\zz)\bigr\|
    \longrightarrow 0.
\]
Therefore $\bar p=\Phi^\star(\bar p;\zz)$. Thus $\bar p$ is a fixed point of the limiting clearing system. By the assumed
uniqueness of the limiting clearing vector, $\bar p=p^\star(\zz)$. Since every convergent subsequence of $\{p^{(t)}(\zz_t)\}_{t\ge 1}$ has the
same limit $p^\star(\zz)$, the full sequence converges:
\[
    t^{-1}p(t\zz_t)
    =
    p^{(t)}(\zz_t)
    \longrightarrow
    p^\star(\zz).
\]

Finally, the aggregate unpaid-obligation loss satisfies
\[
    L(t\zz_t)
    =
    \sum_{j=1}^m
    \left(
        \bar p_j(t\zz_t)-p_j(t\zz_t)
    \right).
\]
Hence
\[
    t^{-1}L(t\zz_t)
    = 
    \sum_{j=1}^m
    \left(
        \bar p_j^{(t)}(\zz_t)-p_j^{(t)}(\zz_t)
    \right)                                                 \to 
    \sum_{j=1}^m
    \left(
        \bar p_j^\star(\zz)-p_j^\star(\zz)
    \right)
    =
    L^\star(\zz).
\]
Therefore $L(\cdot)$ satisfies Assumption~\ref{assume:loss} with
\[
    L^\star(\zz)
    =
    \sum_{j=1}^m
    \left(
        \bar p_j^\star(\zz)-p_j^\star(\zz)
    \right).
\]
This proves the lemma. \qed

\noindent \textbf{Proof of Proposition~\ref{prop:ldp_stress_perspective}:}
Let
\[
    \YY_u := [t(u)]^{-1}\mv{\Lambda}(\XX),
    \qquad
    L_u(\zz):=u^{-1}L(\qq(t(u)\zz)).
\]
Then $\{L(\XX)>u\}=\{L_u(\YY_u)>1\}$. Moreover, by the definition
$\mathcal M_{\delta,u}=\qq(t(u)\mathcal S_\delta^c)$, the event
$\{\XX\in\mathcal M_{\delta,u}\}$ is represented on the standardized scale as
$\{\YY_u\in\mathcal S_\delta^c\}$. Hence
\[
\mathbb P\!\left(\XX\in\mathcal M_{\delta,u}\,\middle|\,L(\XX)>u\right)
=
\frac{
    \mathbb P(\YY_u\in \mathcal S_\delta^c,\ L_u(\YY_u)>1)
}{
    \mathbb P(L_u(\YY_u)>1)
}.
\]
By \cite{deo2025achieving}, Theorem~2,
\[
\lim_{u\to\infty}[t(u)]^{-1}
\log \mathbb P(L_u(\YY_u)>1)
=
-\kappa^\star,
\]
where
\[
\kappa^\star
=
\inf_{\zz\in\Gamma}\varphi^\star(\zz),
\qquad
\Gamma
=
\{\zz:L^\star(\qq^\star \zz^{1/\mv\alpha})\ge 1\}.
\]
For the numerator, define the closed set
\[
    F_\delta := \{\zz:d(\zz,\mathcal S)\ge \delta\}.
\]
Since $\mathcal S_\delta^c\subseteq F_\delta$, we have
\[
    \mathbb P(\YY_u\in \mathcal S_\delta^c,\ L_u(\YY_u)>1)
    \le
    \mathbb P(\YY_u\in F_\delta,\ L_u(\YY_u)>1).
\]
The large deviations upper bound gives
\[
\limsup_{u\to\infty}[t(u)]^{-1}
\log \mathbb P(\YY_u\in \mathcal S_\delta^c,\ L_u(\YY_u)>1)
\le
-
\inf_{\zz\in F_\delta\cap\Gamma}\varphi^\star(\zz).
\]
We next show that the constrained infimum is strictly larger than
$\kappa^\star$. Suppose, to the contrary, that
\[
\inf_{\zz\in F_\delta\cap\Gamma}\varphi^\star(\zz)
=
\kappa^\star.
\]
Then there exists a sequence $\zz_n\in F_\delta\cap\Gamma$ such that
$\varphi^\star(\zz_n)\downarrow\kappa^\star$. Since $\varphi^\star$ is a good
rate function, the sequence eventually lies in a compact sublevel set. Passing
to a subsequence if necessary, let $\zz_n\to\bar\zz$. The sets $F_\delta$ and
$\Gamma$ are closed, so $\bar\zz\in F_\delta\cap\Gamma$. By lower
semicontinuity of $\varphi^\star$,
\[
    \varphi^\star(\bar\zz)\le \kappa^\star.
\]
Since $\bar\zz\in\Gamma$ and $\kappa^\star$ is the minimum of
$\varphi^\star$ over $\Gamma$, equality holds. Hence
$\bar\zz\in\mathcal S$. But $\bar\zz\in F_\delta$ implies
$d(\bar\zz,\mathcal S)\ge\delta$, a contradiction. Therefore
\[
\eta_\delta
:=
\inf_{\zz\in F_\delta\cap\Gamma}\varphi^\star(\zz)
-
\kappa^\star
>0.
\]
Combining the numerator and denominator estimates gives
\[
\begin{aligned}
\limsup_{u\to\infty}[t(u)]^{-1}
\log
\mathbb P\!\left(\XX\in\mathcal M_{\delta,u}\,\middle|\,L(\XX)>u\right)
&\le
-\inf_{\zz\in F_\delta\cap\Gamma}\varphi^\star(\zz)
+\kappa^\star \\
&=
-\eta_\delta
<0.
\end{aligned}
\]
This proves the proposition.\qed

\noindent\textbf{Proof of Proposition~\ref{prop:ldp_transform}:} Note that from \cite{deo2025achieving}, Theorem 1, $\YY_u\in \LDP(t,\varphi^\star)$, and that
$\YY_{s,u}=\tilde{\mv T}_{s,u}(\YY_u)$. To complete the proof, we show that
$\tilde{\mv T}_{s,u}\to \mv f_s$ uniformly on the sets
$\{\zz:\varphi^\star(\zz)\le \gamma\}$ for every $\gamma>0$, where
$f_s(\zz)=s^{\alpha_\star}\zz$. Since $\varphi^\star$ is a good rate function, these sets are compact, so it is
enough to establish uniform convergence on compact sets. The approximate contraction principle
(see \cite{dembo2009large}, Theorem 4.2.23) then yields
\[
\YY_{s,u}=\tilde{\mv T}_{s,u}(\YY_u)\in\LDP(t,\varphi_s^\star),
\]
where $\varphi_s^\star(\zz)=\inf\{\varphi^\star(\yy):f_s(\yy)=\zz\}$. Since $f_s(\yy)=s^{\alpha_\star}\yy$, 
$\varphi_s^\star(\zz)=\varphi^\star(\zz s^{-\alpha_\star})
=s^{-\alpha_\star}\varphi^\star(\zz)$,
where the last equality follows from the homogeneity of $\varphi^\star$
(see \cite{deo2025achieving}, Lemma 2(b)). The rest of the proof shows uniform convergence. 

From \cite{deo2025achieving}, Corollary 1, $[\qq(t(u))]^{-1} \mv T_s(\qq(t(u) \zz)) \to (s^{\alpha_\star} \zz)^{1/\mv \alpha}$ uniformly over compact sets. As a consequence, $ \mv T_s(\qq(t(u) \zz)) = \qq(t(u)) (s^{\alpha_\star} \zz)^{1/\mv \alpha}(1+\mv r_u(\zz))$ where $\sup_{\zz\in K}\|\mv r_u(\zz)\|\to 0$ as $u\to\infty$ for any compact set $K$. In particular $(s^{\alpha_\star} \zz)^{1/\mv \alpha}(1+\mv r_u(\zz))$ stays within a compact set, call it $C$ for all $u$ large enough whenever $\zz\in K$. Since  $ \Lambda_i\in \RV(\alpha_i)$ and is increasing, \cite{deHaan}, Proposition B.1.9 applied coordinate wise implies that $\mv \Lambda\!\left(\qq(t(u))\,(s^{\alpha_\star}\zz)^{1/\mv \alpha}(1+\mv r_u(\zz))\right)
\sim t(u)\, s^{\alpha_\star}\zz$, uniformly over $C$ (since $\mv\Lambda^{\leftarrow} =\qq$). Therefore, we have $\tilde {\mv T}_{s,u}(\zz) \to s^{\alpha_\star} \zz$ uniformly over compact which concludes the proof. \qed

\noindent \textbf{Proof of Theorem~\ref{thm:representative_transform}:} Recall that from  Proposition~\ref{prop:ldp_transform}, $\YY_{s,u}\in \LDP(t,\varphi^\star_s)$. Next, observe that
\[
\{L(\tilde \XX_s)>u\}
=
\{L(\mv T_s(\qq(t(u)\YY_u)))>u\}
=
\{L(\qq(\mv\Lambda(\mv T_s(\qq(t(u)\YY_u)))) )>u\}
=
\{L_u(\YY_{s,u})>1\}.
\]
The third equality follows since $\YY_{s,u} = \tilde{\mv T}_{s,u}(\YY_u)$ where $\tilde{\mv T}_{s,u}(\zz)= [t(u)]^{-1} \mv\Lambda(\mv T_s(\qq(t(u)\zz)))$. Therefore, 
\[
\mathbb P\left(\tilde \XX_s \in \mathcal M_{\delta,u} \mid L(\tilde \XX_s) > u\right) = \mathbb P(\YY_{s,u} \in \mathcal S_{\delta}^c \mid L_u(\YY_{s,u})>1).
\]
Following the proof of Proposition~\ref{prop:ldp_stress_perspective} now yields 
\[
\limsup_{u\to\infty}[t(u)]^{-1}
\log \mathbb P\!\left(\tilde\XX_s \in \mathcal M_{\delta,u} \,\middle|\, L(\tilde\XX_s)>u\right)
\le
-\big(\kappa^\star_s+c_s(\delta)\big)+\kappa^\star_s
=
-c_s(\delta)
\]
where $\kappa_s^\star = \min\{\varphi^\star_s(\zz):\zz\in \Gamma\}$. With $\varphi^\star_s(\zz) = s^{-\alpha_{\star}}\varphi^\star(\zz)$, $\kappa^\star_s = s^{-\alpha_\star}\kappa^\star$. In particular $c_{s}(\delta) > 0$.\qed

\noindent \textbf{Proof of Proposition~\ref{prop:ease_of_learning}:} Note that from Proposition~\ref{prop:ldp_transform}, $\YY_{s,u}\in \LDP(t,\varphi^\star_s)$ and that $\{L(\tilde \XX_s) > u \} = \{L_u(\YY_{s,u}) > 1\}$. Applying the proof of \cite{deo2025achieving}, Theorem 2, with $\varphi^\star_s$ instead of $\varphi^\star$ yields 
\[
\lim_{u\to\infty} [t(u)]^{-1}\log \mathbb P\left(L(\tilde \XX_s) > u \right) = -\inf\{s^{-\alpha_\star} \varphi^\star(\zz): L^\star(\qq^\star\zz^{1/\mv\alpha})\geq 1\} = -s^{-\alpha_\star}\kappa^\star. 
\]
Combining the above asymptotic with the tail probability bound from earlier on, 
\[
\frac{\log (P(L(\XX)>u))} {\log (P(L(\tilde\XX_s ) > u)}\to s^{\alpha_\star} \qed
\]
\bibliographystyle{plainnat}

% AUTHOR: Include your bib file here
\bibliography{ref}
\end{document}